\documentclass[12pt,a4paper,reqno]{amsart}

\pdfoutput=1
\usepackage[margin=0.86in]{geometry}
\usepackage[dvipsnames]{xcolor}
\usepackage{amsthm}
\usepackage{enumitem}
\usepackage{amsmath}
\usepackage{float}
 \usepackage{pdflscape}
      \usepackage{amssymb}
\usepackage{mathrsfs}
\usepackage{tabularx}
\usepackage[bookmarks=false]{hyperref}
    \hypersetup{
         colorlinks   = true,
         citecolor    = RoyalBlue
    }
    \hypersetup{linkcolor=RoyalBlue}
\usepackage{relsize}

\usepackage[utf8]{inputenc}
\usepackage{tikz-cd}
\usetikzlibrary{arrows}
\tikzset{
commutative diagrams/.cd,
arrow style=tikz,
diagrams={>=stealth}}

\usepackage[english]{babel}

\usepackage{lmodern}
\makeatletter
\ifcase \@ptsize \relax
  \newcommand{\miniscule}{\@setfontsize\miniscule{4}{5}}%
\or
\or
  \newcommand{\miniscule}{\@setfontsize\miniscule{5}{6}}%
\fi
\makeatother

\makeatletter
\ifcase \@ptsize \relax
  \newcommand{\nano}{\@setfontsize\miniscule{3.5}{4.5}}%
\or
  \newcommand{\nano}{\@setfontsize\miniscule{4.5}{5.5}}%
\or
  \newcommand{\nano}{\@setfontsize\miniscule{4.5}{5.5}}%
\fi
\makeatother

\newcommand{\balita}{\raisebox{1.8pt}{\text{ \nano$\bullet$\hspace{1.7pt} }}}

\usepackage{mathtools}

\newtheorem{theorem}{Theorem}[section]

\newtheorem{claim}[theorem]{Claim}

\theoremstyle{definition}

\theoremstyle{remark}
\newtheorem{remark}[theorem]{Remark}

 \usepackage{textcomp}
\numberwithin{equation}{section}

\usepackage{graphicx}
 \definecolor{VerdeFH}{HTML}{009374}
\definecolor{BR}{HTML}{002B36}

\newcommand{\eeqref}[1]{eq. \eqref{#1}}

    \definecolor{azulf}{HTML}{0092D2}
\definecolor{azulc}{HTML}{00AEEF}
    \numberwithin{equation}{section}

    \renewcommand{\and}{\hphantom{a}\mbox{and}\hphantom{a}}

    \DeclareMathOperator{\Tr}{Tr}

    \newcommand{\mtr}[1]{\mathrm{#1}}

    \newcommand{\dif}[1]{\mathrm{d}#1}

    \newcommand{\re}{\mathbb{R}}

    \newcommand{\G}{\mathcal{G}}
    \renewcommand{\H}{\mathcal{H}}
    
    \newcommand{\C}{\mathbb{C}}

    \newcommand{\ee}{\mathrm{e}}

    \newcommand{\mtc}[1]{\mathcal{#1}}
    \newcommand{\Z}{\mathbb{Z}}

    \newcommand{\hp}[1]{^{(#1)}}

    \usepackage{booktabs}
\usepackage{colortbl}
\usepackage{xcolor}
\usepackage{graphicx}

\usepackage[skip=.6ex]{subcaption}
  {\csname align*\endcsname}
  {\csname endalign*\endcsname} 
\newcommand{\choosewithout}[2]{\begin{array}{c}
                                #1 \\[-5pt] #2
                               \end{array}
}

   \let\langleb=\langle
   \let\rangleb=\rangle
 \usepackage{mathabx}

\newcommand{\includegraphicsd}[2]{\raisebox{-.415\height}{\includegraphics[width=#1\textwidth]{#2}}}
\newcommand{\includegraphicsdt}[2]{\raisebox{-.3\height}{\includegraphics[width=#1\textwidth]{#2}}}

\newcommand{\includegraphicswextra}[3]{\raisebox{-.#3\height}{\includegraphics[width=#1\textwidth]{#2}}}
 
\newcommand{\itemb}{\item[$\balita$]} 
\hypersetup{urlcolor=RoyalPurple}
\newcommand{\Sint}{S_{\text{int}}}

\definecolor{verdechido}{HTML}{A8DCA8}
\usepackage{afterpage}
\usepackage{enumitem}

\begin{document}
 \author{Carlos I. P\'erez-S\'anchez}
 \address{Faculty of Physics, University of Warsaw 
 \phantom{..-------------------------------------------------------------------}\linebreak
 \phantom{--.-}ul. Pasteura 5, 02-093 Warsaw, Poland \phantom{..............................------------------------------------------------} \linebreak 
  \phantom{--.-}European Union\\ 
 } 

 $\hspace{-.4cm}$\email{cperez@fuw.edu.pl} 
 
\title[Comment on ``phase diagram$\ldots$matrix model with $ABAB$-interaction from FRG''$\quad$]{Comment on \\[3pt]``The phase diagram of the multi-matrix\\  model with ABAB-interaction from\\  functional renormalization''
}

\begin{abstract}{
Recently, [JHEP \textbf{20} 131 (2020)] obtained (a similar, scaled version of) the ($a,b$)-phase diagram derived from the  Kazakov--Zinn-Justin solution of the Hermitian two-matrix model with interactions
    \[-\Tr\Big\{\frac{a}{4} (A^4+B^4)+\frac{b}{2} ABAB\Big\}\,,\]
starting from Functional Renormalization. 
We comment on something unexpected:  the phase diagram of [JHEP \textbf{20} 131 (2020)] is based on a $\beta_b$-function 
that does not have the one-loop structure 
of the Wetterich-Morris Equation. This raises the question of how to reproduce the phase diagram from a set of 
$\beta$-functions that is, in its totality, consistent with Functional Renormalization. A non-minimalist, yet simple truncation that could lead to the phase diagram is provided.  Additionally, we identify the ensemble for which the result of \textit{op. cit.} would be entirely correct.}
\end{abstract}

\maketitle
\section{Main claim and organization}\label{sec:claim}

We prove the following: For a Hermitian 
two-matrix model including the `$ABAB$-interaction' vertex \[\frac b2 \Tr \,(ABAB)=\includegraphicsd{.115}{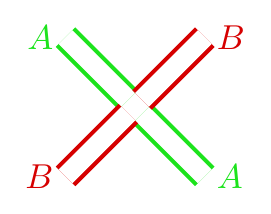}\]
the only one-loop, one-particle irreducible (1PI) diagram 
of order $b^2$ 
that has connected external leg structure---that is, such that it contributes to a connected-boundary correlation function---is
  \begin{align}
                        \label{graph}\includegraphicsd{.185}{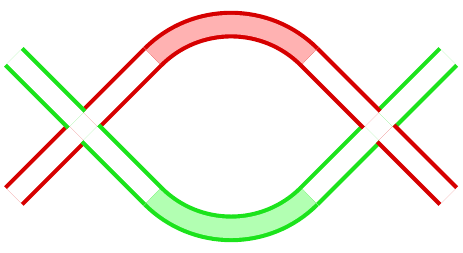}   \, 
                        \end{align}
                        
Its external leg structure is the cyclic word $ABBA$, represented by  \includegraphicsd{.115}{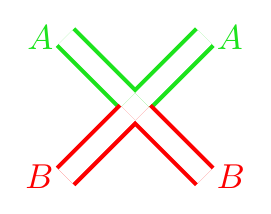}. 
This implies that the $\beta_b$-function given in \cite[\text{Eq. 3.41}]{ABAB}, namely 
 
  \begin{align}\label{nichtsoganzrichtiges_Beta}
\beta_{b}\stackrel{}{=}(\eta_A+\eta_B+1 )\times b-
  \underbrace{\frac{2}{5}\big[ (5-\eta_A)+ (5-\eta_B)\big] \times\,
    b^2}_{\text{\tiny 1-loop structure $\Rightarrow$ coeff. of $b^2$
   must be   0}} ,\end{align}
    does not have the one-loop structure of the Wetterich-Morris equation 
      in Functional Renormalization\footnote{Here $\eta_A$ and $\eta_B$ are the anomalous dimensions
  for the matrices $A$ and $B$, but this is irrelevant, since the
  coefficient of $b^2$ should identically vanish.} (this is relevant, 
  since the phase diagram \cite[Fig. 1]{{ABAB}} relies only on
  $\beta_b$).  For the quadratic
term $b^2$ to be present in this equation, the graph \eqref{graph} would need
to have an external $ABAB$-structure. In other words, for \eeqref{nichtsoganzrichtiges_Beta} to hold, after `filling the
loop' in \eqref{graph} and shrinking the disk to a point, the remaining graph should read, just like the vertex, 
\[
  \includegraphicsd{.115}{Listones/Feo_ABAB}\,\]
Different infrared regulators might lead to different
coefficients (containing non-perturbative information), but the
one-loop structure in Functional Renormalization should be evident;
this means that the coefficient of $b^2$ in $\beta_b$ must vanish.
 
Next, in Section \ref{sec:details}, these statements are presented in detail in Claims \ref{thm:uno},
\ref{thm:dos} and \ref{thm:tres} (at the risk of being redundant) and proven. 
A short (albeit, rich in examples) user's guide to colored ribbon graphs  (Section \ref{sec:termo}) prepares
the core of this comment (Section \ref{sec:proofs}). 
In Section \ref{sec:proposal} we further propose a more generous truncation to obtain
the phase diagram. We compute in the large-$N$ limit without further notice.

\section{Context, definitions, examples and proofs}\label{sec:details}

The two-matrix model in question has the following partition function\footnote{We respect here the
  notation of \cite{ABAB}, the exception being renaming their
  $(\alpha,\beta)$ to $ (a,b)$ as to avoid `$\beta_\beta$'.}
\begin{subequations} \label{model}
\begin{align}\label{partfunct}
\mathcal Z= C_N \iint_{ \H_N\times \H_N} \ee^{-\frac N2 (\Tr A^2+\Tr B^2) - N \Sint(A,B)}  \dif A\, \dif  B  
\end{align}
with $C_N$ a normalization constant and $\dif A $ and $\dif B$ both the Lebesgue measure on
the space $\mathcal H_N$ of Hermitian $N\times N$  matrices, and interaction
\begin{align}\label{bare}
\Sint (A,B)= -\frac a4 \Tr (A^4 +   B^4  ) - \frac{b}{2} \Tr (ABAB) \,.
\end{align}
In 1999, the  $ABAB$\textit{-model} \eqref{model} was exactly solved by Kazakov and (P.) Zinn-Justin, who presented in 
\cite[Fig. 4]{KazakovABAB} a phase diagram of right-angled trapezoidal form 
for the couplings ($a,b$), called there ($\alpha,\beta$) as well as in \cite{ABAB}. 
A consistent phase diagram with trapezoidal \textit{form} (i.e. predicting
$\approx 1/10$ for both critical exact $a_\star=b_\star=1/4 \pi$ values) is 
one of the main results
presented in \cite{ABAB}, who addressed the model \eqref{model} using
Functional Renormalization.
The phase diagram \cite[Fig. 1]{ABAB} obtained from
$(\beta_a,\beta_b)$ follows from a correct expression for $\beta_a$
but also from  
$ \beta_{b}\stackrel{}{=}\big(\eta_A+\eta_B+1
\big)b-{\frac{2}{5}\left[ (5-\eta_A)+ (5-\eta_B)\right] b^2}$, which is \cite[Eq. 3.41]{ABAB}.  We
prove here that the $\beta_b$-function is incompatible with the
well-known\footnote{The statement (revisited below) follows from the form of the
  Wetterich-Morris equation and appears e.g. in
  \cite[Sec. 3.1 \S 1]{ABAB} `As the full propagator enters,
  non-perturbative physics is captured, despite the one-loop
  structure'. It also appears in several introductory texts to Functional Renormalization, e.g. \cite{Gies}.}  one-loop structure \cite{Berges:2000ew} of Wetterich-Morris Functional
Renormalization Group Equation \cite{Wetterich,Morris}.  While the
behavior $\beta_g \sim g^2 + \ldots $ is indeed common for other
quartic operators $g\mathcal O$, this does not happen for  $b \Tr(ABAB)$. 
The rest of the section introduces the 
terminology (Section \ref{sec:termo})
and proves in detail the claims (Section \ref{sec:proofs}).

\subsection{Colored ribbon graphs in multi-matrix models by example}\label{sec:termo}
Feynman graphs turn out to be useful also for `non-perturbative \cite{Berges:2000ew}
renormalization'.  
For sake of accessibility to a broader readership,  
we provide in this section an (incomplete) user's guide to graphs in multi-matrix models. 

The representation of 
the integrals of matrix models using \textit{ribbon graphs} 
(or \textit{fat graphs}), 
famous due to 't Hooft \cite{thooft}, is of paramount importance both in physics and mathematics; applications are also worth mentioning
\cite{GenChordDiags,Andersen:2013tsa}  (see in particular its relation to 
discrete surfaces or \textit{maps} studied by
Brezin-Itzykson-Parisi-Zuber \cite{Brezin:1977sv}). The theory of ribbon graphs
can be formulated in an extremely precise way \cite{Penner:1988cza},
but for the purpose of this comment, the most important feature is that their
vertices have a cyclic ordering---this is typically depicted with the aid of a disk 
with some thick strips (\textit{half-edges})  disjointly attached to it. Each 
strip represents a matrix, and these are adhered (in our convention, clockwise) to the vertex, as in the following picture:
\[
g \Tr(AACDBCDB) \leftrightarrow
\includegraphicsd{.17}{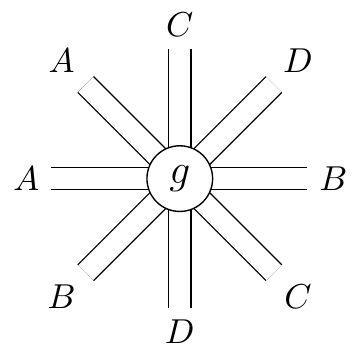}\,\,\, 
\]
 where $A,B,C,D$ are matrices of the same ensemble\footnote{This
 vertex will not be used, this example only aims at explaining
 the concept with some more clarity.}.
Sometimes the disk is omitted, as we often do below,
and the coupling constant (here $g$), too. The rotation of the vertex is conform with the cyclicity of the trace, but one is not allowed to reflect the picture.

This way, ribbon vertices, unlike ordinary ones,
are sensitive to non-cyclic reorderings of the half-edges (see e.g. \cite[Fig. 1 and
eq. (18)]{cips}). The representation of the interaction \eqref{bare}
\end{subequations}
in terms of ribbon (or fat) vertices reads\footnote{In case of color-blindness,
we are calling \textit{green} the lighter lines and \textit{red} the darker ones. These represent $A$ and $B$, respectively. }: \begin{align}
- \Sint(A,B)=   
\includegraphicsd{.115}{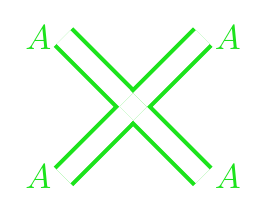}
+
\includegraphicsd{.115}{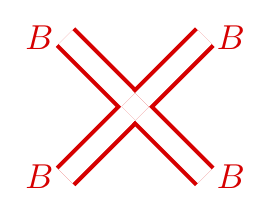}
+ 
\includegraphicsd{.115}{Listones/Feo_ABAB} \,
\end{align}
Since the above graphical representation will be used only as a
cross-check, we ignore the symmetry factors (strictly, we should put a
root on one edge of each interaction vertex) and also absorb the couplings in the vertices.  
The cyclic ordering means, in particular, that 
\begin{align} \label{obvious}
  \includegraphicsd{.1015}{Listones/Feo_ABAB} \neq
  \includegraphicsd{.1015}{Listones/Feo_RibbonsABBA}, \mbox{ which
    faithfully represents the obvious: $\Tr (ABAB) \equiv\!\!\!\!\! \slash\,\, \Tr(ABBA)$}.
\end{align} 
 
\textit{Edges} are also fat (double lines)
and consist of pairings of half-edges (which for Hermitian
 matrix models cannot have net `twists' \raisebox{-.145\height}{\includegraphics[width=22pt]{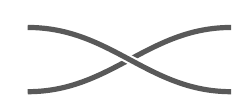}}).
In order to emphasize that 
these arise from propagators, we shade the edges. While
for a one-matrix model all edges are equal---for instance,
the next graph arising in a quartic one-matrix  model,
\begin{align}\label{grafo_cuartico}
\includegraphicsd{.19}{Listones/Feynman_Ribbons_quartic}
\end{align}
---the new feature in multi-matrix models the coloring
of the edges. In the $ABAB$-model 
(which contains the $A^4$ and $B^4$ interactions)
the graph \eqref{grafo_cuartico} is possible 
in green (implying only the $A$ matrix) or in red (only $B$), 
but a coloring of the graph \eqref{grafo_cuartico} in such
a way that it contains (at least one occurrence of) the $ABAB$-\textit{vertex} is not possible. An example of
a ribbon graph with two such vertices is  
\begin{align}\label{topo}
\includegraphicsd{.38}{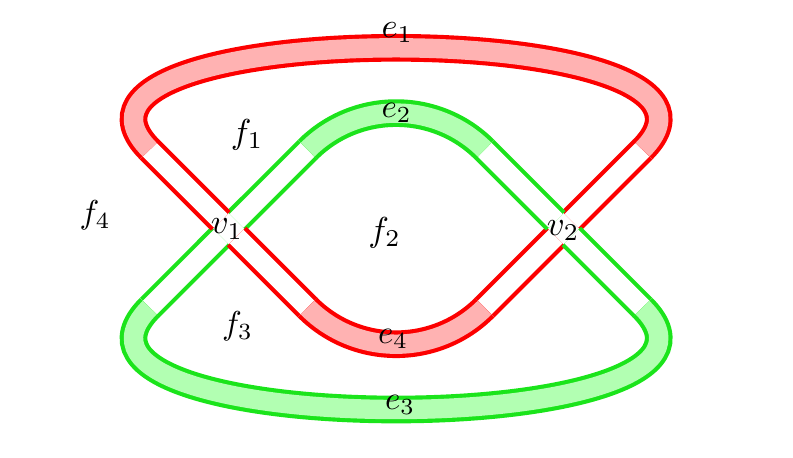}
\end{align}
in which also the vertices $v_i$, edges or propagators\footnote{The propagators are shaded and they respect the color
  of the associated matrix.} $e_i$, 
and faces $f_i$ are depicted (a \textit{face} is a
boundary component of the ribbon graph; here we have three bounded faces, $f_1,f_2,f_3$, 
and one unbounded,  $f_4$). The colored ribbon
graphs of multi-matrix models, just as their uncolored version, have a topology
determined by the Euler number $\chi$ (where $\chi=\#\,$vertices $ -\#\,$edges$\, +\#\,$faces, in this case clearly $\chi=2$,
corresponding to a spherical topology; this turns out to be important, since the 
scaling of the graph amplitudes with the matrix size $N$ is $\sim N^\chi$).
If we would have an octic interaction vertex
$g \Tr(ABABABAB)$, one of the (several) possible diagrams is 
\begin{align} \label{genustwo}
\includegraphicsd{.3384}{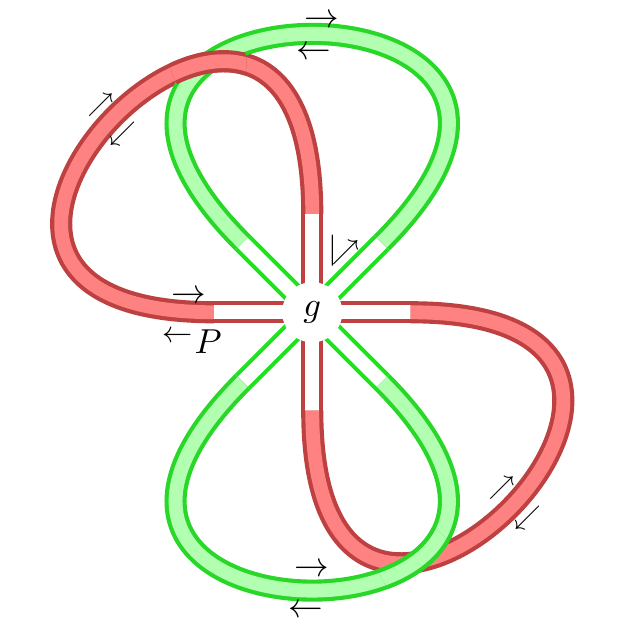}
\end{align}
One can verify, by starting at any point of the 
boundary of the diagram---say, at the point $P$ in the picture---and by following the arrows, 
that when one comes back to $P$, one has already visited the whole boundary once.
Therefore, this fat graph has a single face.
In that picture, the arrow with a 45$^\circ$ angle $|\hspace{-6.1pt}\nearrow$ emphasizes 
the criterion to travel along the circle that defines the boundary of the face: when one arrives at the disk (marked with the coupling $g$) 
one picks the closest single-line, regardless of the color, for the cyclic order at the  
vertex determines precisely, which goes next.
The counting of vertices (one), faces (one) and edges (four) exhibits its non-planarity (i.e. it cannot be draw of a sphere;
the best one can do is draw it without intersections 
on a surface of genus 2) but below we will find only planar diagrams.

So far, all examples we presented are \textit{vacuum graphs}. We now consider
also graphs having half-edges that are not contracted, as in
\begin{align}\label{brokenfaces}
\includegraphicsd{.34}{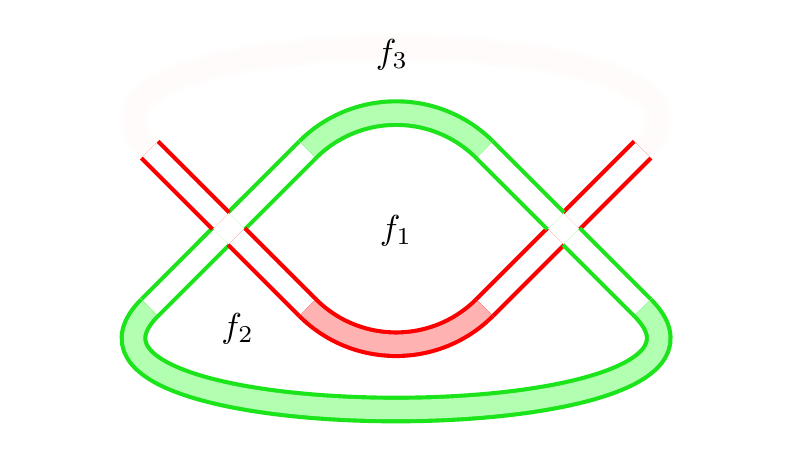}
\end{align}
A face of a ribbon graph  is said to be
\textit{unbroken} if no (uncontracted) half-edges are incident to it; the face is otherwise said
to be \textit{broken}, precisely by the incident half-edges.  
(Thus, vacuum graphs can only have unbroken faces.)
In the graph \eqref{brokenfaces},
the faces $f_1$ and $f_2$ are unbroken, while 
the unbounded face $f_3$ is broken by the two red half-edges pointing northeast and northwest.
\par 
Given a (colored) ribbon graph one can forget the cyclic ordering 
of its vertices (and if present, the coloring of its edges) and thus obtain a graph in the ordinary sense.
For instance,
\begin{align}
\raisebox{-.370\height}{
 \includegraphics[width=69pt]{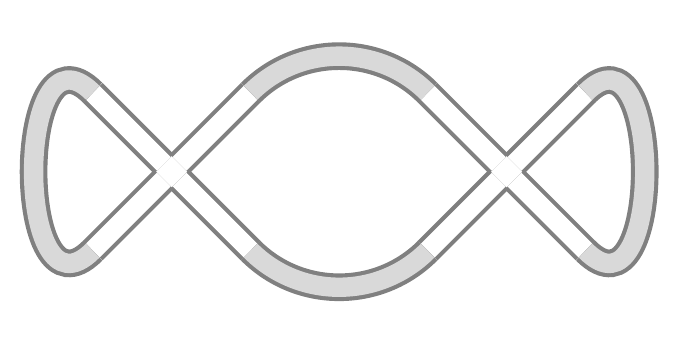}}\,\,
\mapsto \,\,
\raisebox{-.36\height}{\includegraphics[width=39pt]{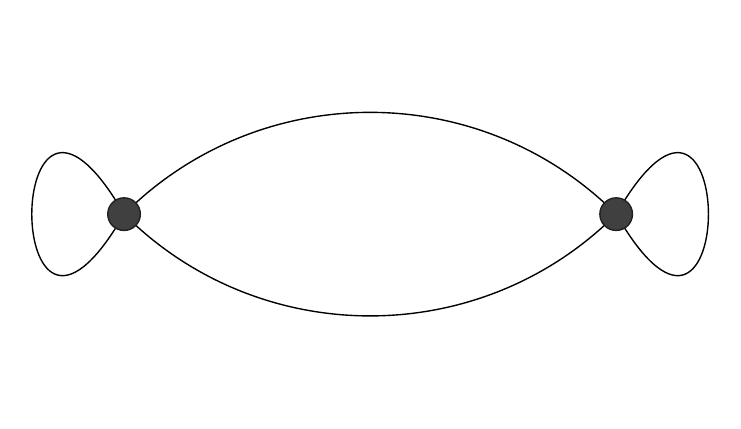}}
\end{align}
and 
\begin{align}
\raisebox{-.30\height}{
 \includegraphics[width=40pt]{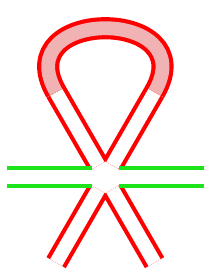}}\,\,
\mapsto \,\,
\raisebox{-.36\height}{\includegraphics[width=39pt]{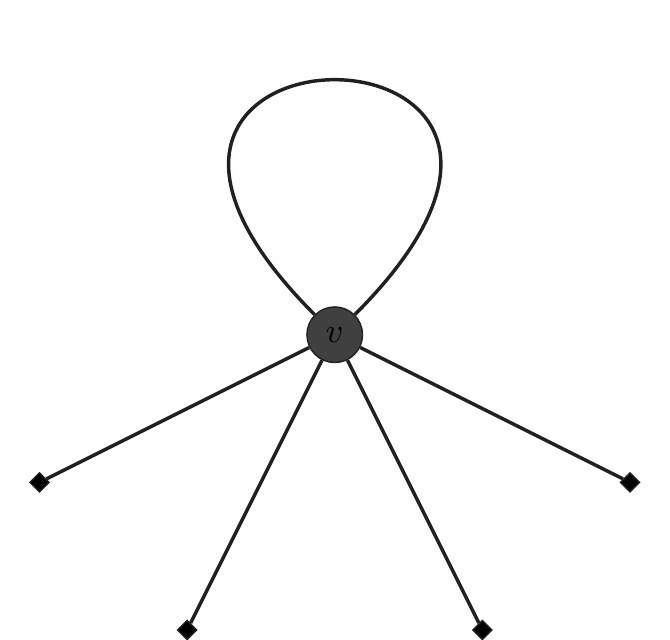}}
\label{c}
\end{align}
Notice that to the uncontracted half-edges of the ribbon graphs
one associates leaf (degree one vertex). 
A connected (colored) ribbon graph is said to have a \textit{one-loop}
structure if its underlying ordinary graph has a one-loop structure, that is,
if the latter has a first Betti number\footnote{The first Betti number of a connected ordinary graph $G$ is $b_1(G)= \# \text{edges of }G -\# \text{vertices of }G + 1$, so `one-loop' means that the number of edges equals the number of vertices. Notice
that the `external edges' are attached to a vertex. This
definition of $b_1$ also matches other conventions,
where external edges have no vertices attached, but then it is given by
 $b_1(G)=\# \text{internal edges of }G -\# \text{vertices of }G + 1$ 
} equal to 1; alternatively,
if it has one independent cycle.
Therefore \eqref{c} above has a one-loop structure, but neither of the following has it:
\allowdisplaybreaks[3]%
\begin{subequations}%
 \begin{align}\label{loopsnoloopsA}
 \raisebox{-.36\height}
{\includegraphics[width=109pt]{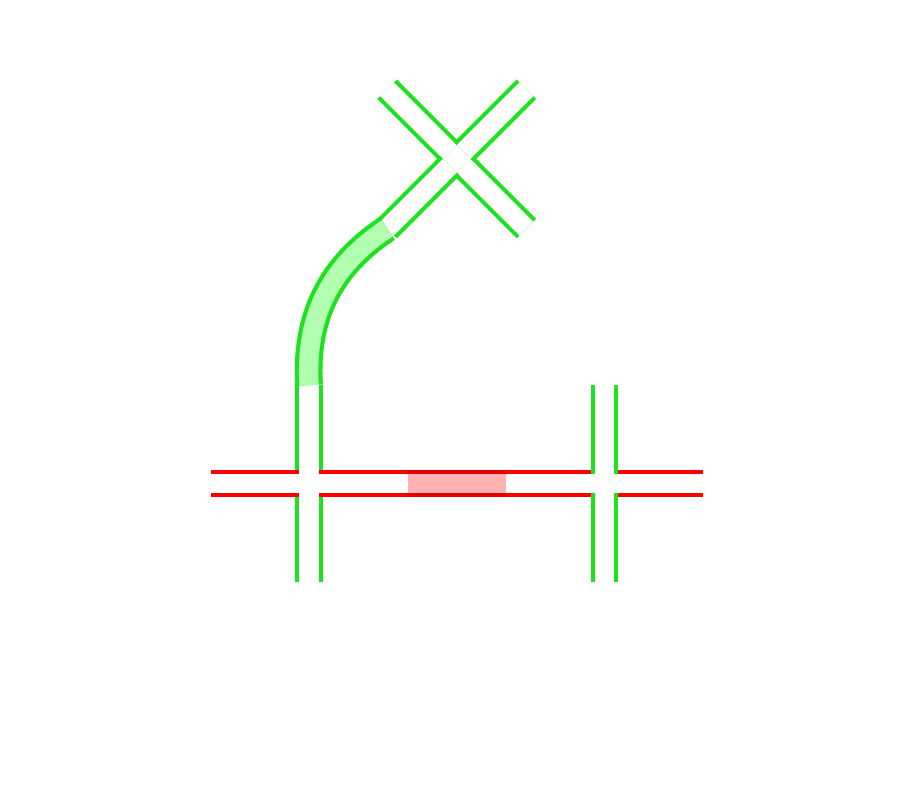}} \hspace{-20pt}&\mapsto 
\quad\raisebox{-.13\height}
{\includegraphics[width=49pt]{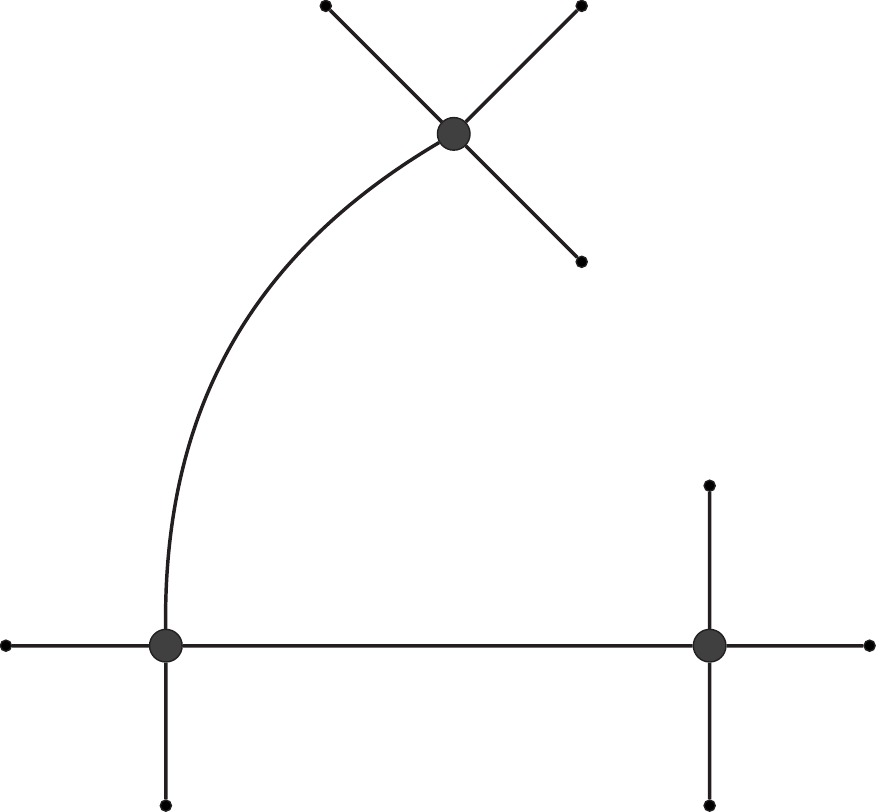}}
\\\label{loopsnoloopsB}
\raisebox{-.36\height}
{\includegraphics[width=109pt]{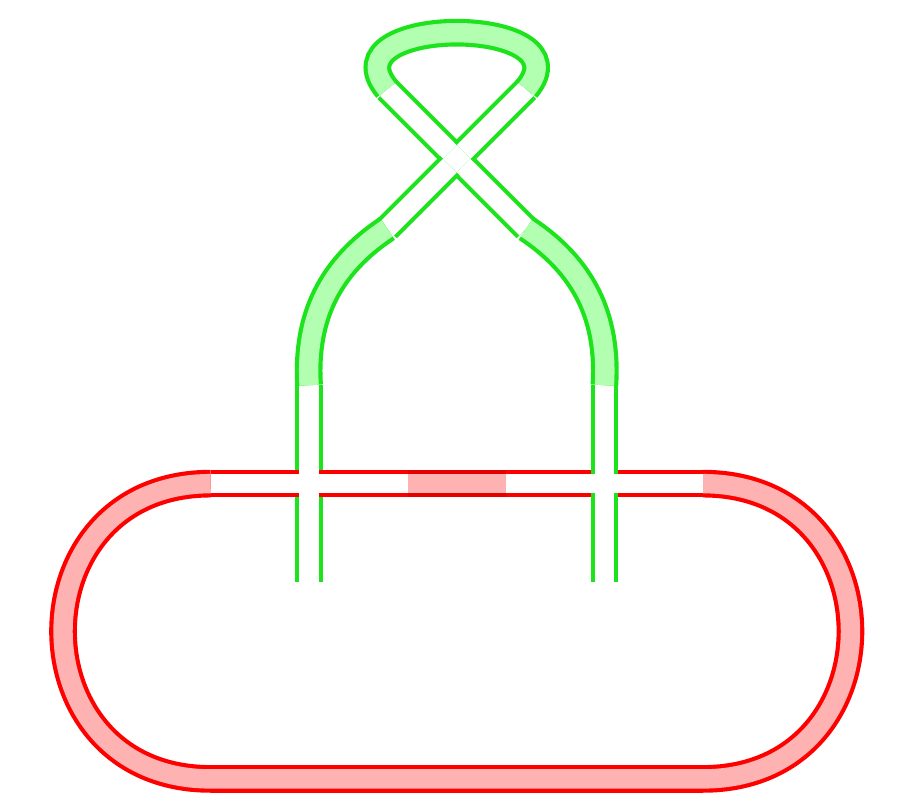}}
&\mapsto \quad
\raisebox{-.13\height}
{\includegraphics[width=40pt]{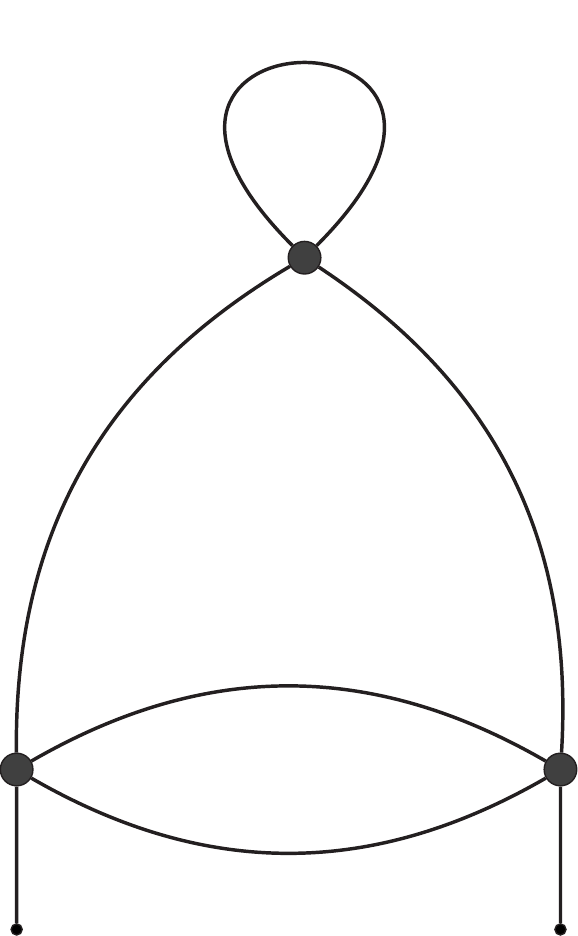}}
\end{align}
\label{loopsnoloops}%
\end{subequations}%
for the upper ribbon graph has no loops and the other has three.
\par 
A connected (colored) ribbon graph is \textit{one-particle irreducible} (1PI) if its underlying ordinary graph is 1PI. 
An ordinary 
graph is 1PI if it is neither a \textit{tree}---a graph for which 
there is a unique path between any two given vertices---nor it can be disconnected by removing exactly  one edge.
The (ordinary) graph in \eqref{loopsnoloopsA}
is a tree \textit{and} it can be disconnected by cutting 
either the straight or the curved edge. Therefore
the fat graph in \eqref{loopsnoloopsA} not 1PI.
On the other hand, the ordinary graph in 
\eqref{loopsnoloopsB} is neither a tree (for there exist
vertices connected by more than one path) and if one 
removes any edge, it remains connected. Therefore
its primitive ribbon graph in the left of \eqref{loopsnoloopsB} is 1PI. 
The next example is not at tree,
\begin{align}
 \raisebox{-.36\height}
{\includegraphics[width=109pt]{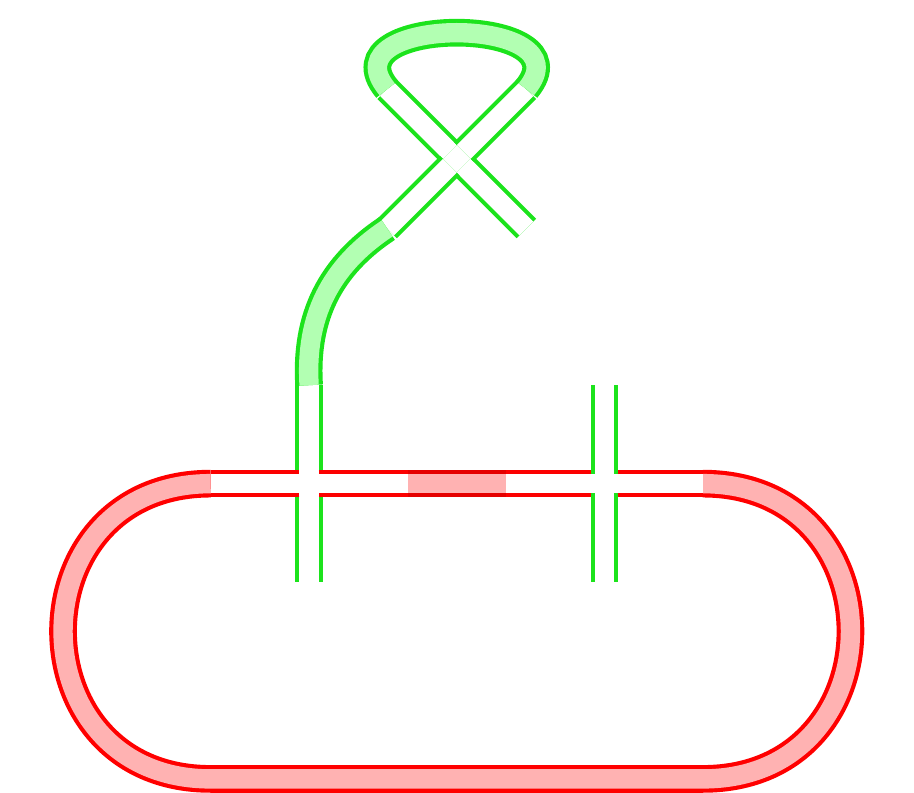}}
\end{align}
 but
removing one 
propagator (the green one in position \raisebox{-8pt}{\tikz{\node[rotate=-45,green!80!black]{$\hspace{-8pt} (\!($};}}\hspace{-4pt}) disconnects it, so it is not 1PI.

We now introduce the last concept. The 
\textit{external leg structure} of a ribbon graph with one 
broken face is obtained by reading off clockwise 
the cyclic word formed by the matrices (associated with 
the half-edges) that break that face, going around the boundary-loop exactly once. 
This process is known in 
the matrix field theory literature \cite{GW12} (illustrated in \cite[Sec. 5]{fullward})
and a generalization to multi-matrix models requires to additionally
list the half-edges respecting the coloring. We illustrate this concept,
reusing graphs previously drawn:
\begin{enumerate}
 \itemb Concerning the graph \eqref{brokenfaces}: it has an 
external leg structure $BB$, since going along the $f_3$ face,
one meets twice a red line.
\itemb The graph \eqref{c} has external leg structure $ABBA$
\itemb The ribbon graph in \eqref{loopsnoloopsA}
has external leg structure $AAABAABA$ 
\itemb That in \eqref{loopsnoloopsB} has external leg structure $AA$
\end{enumerate}
An external leg structure determines in a natural way a new 
interaction vertex of matrix models 
just by `taking its trace'. The 
cyclicity of the external leg structure yields well-definedness. In the 
list of external leg structures of 1 through 4, these correspond,  
to 1. $\Tr(B^2)$, 
2.  $\Tr(ABBA)$, and  3. $\Tr (AAABAABA)$ and 4. $\Tr(A^2)$.
The next example probably explains why we restricted ourselves to graphs with a single broken face,
\begin{align} \label{broken}
 \includegraphicsd{.13}{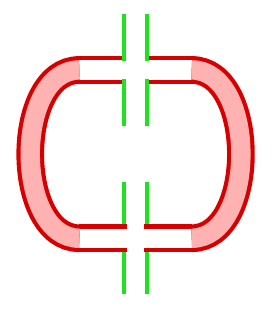}
\end{align}
 The face inside the red loop yields $A^2$; the same from the outer face. 
 Thus, the external leg structure of \eqref{broken} is 
the disjoint union of $A^2$ with $A^2$.
Graphs with more than one broken face will not appear below, since these lead to multi-traces,
in the case of  \eqref{broken} to $\Tr(A^2)\times \Tr (A^2)$, and these multi-trace interactions are not
consider in the article we are commenting (but are treated in \cite{FRGEmultimatrix}). But the idea of a more general setting is depicted in Figure \ref{fig:ExternalLegStr}. \par

We are done with the terminology. In the following section we prove the claims.

\subsection{Proofs}\label{sec:proofs}
In the Functional Renormalization Group parlance, one says that the RG-flow generates the interaction vertices
that arise from the external leg structure,
\textit{whenever these come from a one-loop graph}. The one-loop condition  
is a consequence of the `supertrace' present in Wetterich Equation\footnote{Without giving 
a full description (for that we refer to \cite{FRGEmultimatrix}) we recall
that $R_N$ is the infrared regulator;
that $\Gamma\hp2_N$ is the Hessian of the interpolating 
effective action $\Gamma_N$; and that $t=\log N$ is the 
RG-time.} $\partial_t \Gamma_N= \frac12 \mtr{STr}\, \big\{ \partial_t R_N /[R_N+ \Gamma\hp2_N] \big\} $. 
This makes
some of the graphs listed above uninteresting
from the viewpoint of renormalization, as they do not have a
one-loop structure. Those which do, 
have also some drawbacks: \eqref{c} has a one-loop
structure, and is a graph generated by the RG-flow in the $ABAB$-model,
but that sextic vertex is set to zero in the truncation that \cite{ABAB} considers.
On the other hand, the graph \eqref{broken} is a Feynman graph containing
only vertices of the $ABAB$-model,
but it leads to a double-trace, and so on. \par 
This hopefully slowly starts to convince the reader that the 1PI and the one-loop conditions 
heavily restrict the graphs that play a role in the FRG; this is the flavor of the proofs below, where in fact, we  obtain certain kind of uniqueness.
Consider the following graph:
\begin{align} \label{unbroken}
\mtc G:=\includegraphicsd{.185}{Listones/Feo_RibbonsABAB_ABAB90}\,\,
\end{align}
It has a broken face and an unbroken one,
and when one glues a disk along the boundary 
this happens:

 \begin{figure} 
   \centering
   \includegraphics[width=.40\textwidth]{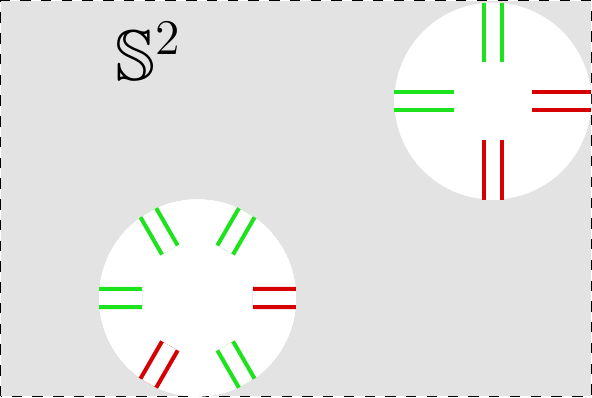}
   \caption{The gray part is an abstract representation of a multi-matrix
model planar graph (after `vulcanizing' the unbroken faces and forgetting everything but broken faces) with two faces broken by uncontracted
half-edges. These boundaries determine the external leg structure
decorating the cylinder (the two broken faces are interpreted as excised discs from $\mathbb S^2$ here). The 4-leg boundary in
the upper right corner yields the cyclic word $ABBA $ (or $AABB$,
or$\ldots$); the 6-legged one is $AABABA$ (or $AAABAB,\ldots$). The external leg structure is
the disjoint union of these words.\label{fig:ExternalLegStr}}
\end{figure} 
\begin{claim}\label{thm:uno}
  The external leg structure of the graph $\mtc G$ given by
  \eqref{unbroken} is $ABBA$, i.e.
\[\includegraphicsd{.125}{Listones/Feo_RibbonsABBA}\]
\end{claim}
\begin{proof}
In $\G$ one has a single broken face (it can be recognized  
in the drawing   \eqref{unbroken} as the face that is unbounded). 
Starting at any point one reads off
at the boundary of such face 
the word $ABBA$. One could also read off $BBAA$, $AABB$ or $BAAB$, depending on where one chooses 
to start. Yet, nothing changes, since the external leg structure is cyclic
by definition.
\end{proof}

\begin{claim}\label{thm:dos} The unique $(\hspace{-.61pt}$connected$\hspace{.81pt})$ one-loop 1PI graph of order $b^2$ having a
  connected external leg structure  is the graph $\mtc G$ defined by \eqref{unbroken}.
\end{claim}
In this statement the connectedness of the external leg structure 
means that $\G$ has a single broken face.
\begin{proof}
  By assumption, the graph has two interaction $ABAB$-vertices; let us
  name $V_1$ and $V_2$ the two copies.  The 1PI-assumption constrains
  the $p$ propagators implied in the graph to $p>1$ (indeed,  
  since the graph should be connected, and has two interaction
  vertices, a propagator should connect these. 
  Should the graph have a single propagator, then removing 
  it would yield a disconnected graph contradicting the 1PI assumption). On the other hand, the
  one-loop condition implies $p<3$ (if $p\geq 3$ then more loops are formed).  Again, since the graph in question is
  1PI, the $p=2$ propagators implied in the graph connect $V_1$ with
  $V_2$.  If they connect two equal colors, we get a disconnected
  external leg structure, i.e. either the graph \eqref{broken} or its
  $A\leftrightarrow B$ (i.e. green $\leftrightarrow$ red) version. Therefore, by assumption, the two propagators
  connecting $V_1$ with $V_2$ must have different color. The only
  such graph having also a connected external leg structure is $\G$ defined
by \eqref{unbroken}.
\end{proof}

It is a Quantum Field Theory folklore result (see e.g. \cite{Borcherds:2002cy} and \cite[Sec. 3.3]{ConnesMarcolli}) that the effective action is the generating functional of 1PI graphs. 
Since neither the cyclicity of the vertices\footnote{The cyclicity of the 
vertices is of course a key feature for the power counting $N^\chi$,
but this and one-particle irreducibility are independent properties.} nor the coloring of thick 
edges has influence on the 1PI property,
the effective action in multi-matrix models 
generates 1PI colored ribbon graphs; thus, that  
folklore result remains true in this setting. 
Further, the Functional Renormalization Equation
governs the interpolating effective action,
hence the 1PI condition is imposed from the outset\footnote{The exception is
the vertices appearing as linear terms in the $\beta$-functions.}.
We have, in view of this:

\begin{claim}\label{thm:tres}
  For the particular operator  $-\frac{b}{2} \Tr(ABAB)$, a  quadratic term
  $b^2$ in $\beta_{b} $ is not possible in Functional Renormalization $($in other words
 the coefficient of $b^2$ in $\beta_b$ vanishes,  or in notation, $[b^2]\beta_b=0$$)$.
\end{claim}

\begin{proof}  Wetterich-Morris Equation imposes the one-loop
  structure on any non-linear term in the coupling constants (i.e.$\,\,$on any non-vertex) appearing in each $\beta$-function.  For the
  $\beta_b$-function, in particular, a second condition is that the
   external leg structure must be $ABAB$.  These two conditions
  are mutually exclusive. Indeed, by     Claim \ref{thm:dos}, the only such
  $\mathcal O(b^2)$ graph is $\G$, which, by Claim \ref{thm:uno}, has an
  external leg structure $ABBA$.\end{proof}
\noindent

\begin{remark} Some closely related, but not essential points: 

\begin{itemize} 
  \itemb The correlation functions of matrix \cite{GW12} and tensor
  field theories are indexed by \textit{boundary graphs}
  \cite{Bonzom:2014oua}.  The terminology makes sense graph theoretically but
  also geometrically in both the matrix and tensor field
  \cite{fullward} contexts. In the case of matrix models, boundary
  graphs coincide with what we call here \textit{external leg
    structure}. In matrix models,
  there are as many $\beta$-functions as correlation functions,
  hence the importance of the external leg
    structure. The map defined by `taking the boundary'
    seems also to play a role in other renormalization theories, like Connes-Kreimer Hopf algebra approach \cite{ConnesKreimer}. 
   In that theory for matrices (related construction appears in \cite{TanasaVignes}) taking the boundary seems to be the residue map in 
   terms of which one can define the coproduct of the Hopf algebra 
   (also true \cite{Raasakka:2013kaa} for the  Ben Geloun-Rivasseau tensor field theory \cite{4renorm}).
 
  \itemb Other graphs might appear for real symmetric matrices, but
  the ribbons corresponding to Hermitian matrices remain untwisted,
  which played a role in the uniqueness of the graph above. It is not
  exaggerated to stress the reason for this rigidity, which is
  explained by the difference in the propagators of matrix models in
  ensembles $\{ M\in M_N(\mathbb K) \mid M^\dagger = M\}$ for
  different fields $\mathbb K$, i.e.
\begin{align}
  \langleb M_{ij} M_{kl}\rangleb_{\mathbb K}= \frac{1}{N}\begin{cases}
    \displaystyle\frac{1}{2}  \Big( \choosewithout{i}{j}\! \! \includegraphicsdt{.1}{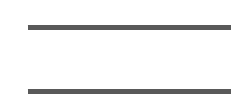} \choosewithout{l}{k}
    +   \choosewithout{i}{j} \!\! \includegraphicsdt{.1}{Listones/mismatch}  \choosewithout{l}{k} \Big) & \mathbb K= \re \,\,(\text{symmetric real})
    \\[2ex] \phantom{\displaystyle\frac{1}{2}  \bigg(}
    \choosewithout{i}{j} \!\!\includegraphicsdt{.1}{Listones/match} \choosewithout{l}{k}& \mathbb K= \C \,\,(\text{Hermitian})
\end{cases}
\label{propagators}
\end{align}

\end{itemize}

\end{remark}
 
\section{A proposal to obtain the Kazakov--Zinn-Justin phase diagram}\label{sec:proposal}
We recall that in \cite{ABAB} the truncation is minimal. Thus, the running operators are also those in the bare
action in eqs. \eqref{partfunct}-\eqref{bare}. \par
In that truncation \cite{ABAB}, if one computes correctly, $\beta_b \sim b $ (the value of the vertex); this is not wrong, but only says that such truncation threw away useful information.
In order not to `waste' the $b^2$ term, we propose to add
the operator that captures it, so the new effective action reads:
\begin{align}\label{newtruncation}
 \Gamma_N[A,B]=\frac{Z}{2}\Tr (A^2+B^2) -\frac{\bar{a}}{4} \Tr (A^4+B^4) -\frac{\bar{b}}{2}\Tr (ABAB)-\bar c \Tr(ABBA)\,,
 \end{align}
  where bar on quantities means \textit{unrenormalized} and
$Z$ is the (now common) wave function renormalization. 
 The $ABBA$-operator also respects the original $\Z_2$-symmetries:
 \begin{align}\label{symmetries}
 (A,B)\mapsto (B,A), \quad (A,B)\mapsto (-A,B)\,, \quad (A,B)\mapsto (A,-B)\,. 
 \end{align}
 Now, $b^2$ does appear, but in the $\beta_c$-function. In fact $\beta_c\sim b^2 + c^2 + 2ac$, ignoring
 the contribution ($\propto \,c$) of the vertex.
This relation was obtained with the (`coordinate-free') method
presented in \cite[around Corollary 4.2]{FRGEmultimatrix}.
Removing the symmetry condition
($g_{AAAA}=a=g_{BBBB}$) we initially imposed on the couplings for $A^4$ and $ B^4$, and writing in full $g_{w}$ for the coupling of
$\Tr\{w(A,B)\}$, being $w$ a cyclic word in $A,B$, one has\footnote{Or
  $\beta{(g_{ABBA})} - g_{ABBA} (\eta_A+\eta_B +1) \sim {g_{AAAA}
    \times g_{ABBA}} + g_{BBBB} \times g_{ABBA} + (g_{ABAB})^2 +
  {(g_{ABBA})^2}$ in case that colors distract the reader.}
\begin{subequations} 
\begin{align} 
  \qquad \qquad \beta{(g_{\color{green!67!black}A\color{red!84!black}BB\color{green!67!black}A})} - g_{\color{green!67!black}A\color{red!84!black}BB\color{green!67!black}A} (2\eta +1) &  \sim  \nonumber 
                                                                                                                                                                                                  {g_{\color{green!67!black} AAAA} \times  g_{\color{green!67!black}A\color{red!84!black}BB\color{green!67!black}A}}  +  g_{\color{red!80!black}BBBB} \times g_{\color{green!67!black}A\color{red!84!black}BB\color{green!67!black}A}   \\ &\!\qquad+  (g_{\color{green!67!black}A\color{red!84!black}B\color{green!67!black}A\color{red!84!black}B})^2 +     {(g_{\color{green!67!black}A\color{red!84!black}BB\color{green!67!black}A})^2}\,,  \label{Moja}  \end{align} 
                                                                                                                                                                                                obtained without the use of graphs. The missing coefficients (hidden in $\sim$ and implying the anomalous dimension $\eta= -N\partial_N Z$) are regulator-dependent and contain non-perturbative information, but the essential point is that now each $\beta$-function has a transparent one-loop structure; to wit, the RHS of \eqref{Moja} corresponds (respecting the order in that sum\footnote{Of course, the rotation of these graphs plays no role, since only the cyclic order matters. We display some vertically, and others horizontally for sake of the reader's comfort: this way, she or he can start reading the word from the upper left corner anticlockwise and this order will coincide with the order of the word in the corresponding coupling constant.}) with
\begin{align} 
\includegraphicswextra{.1}{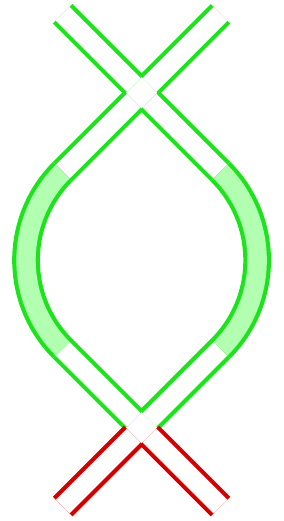}{412} \,\,\, +\,\,\,
\includegraphicswextra{.1}{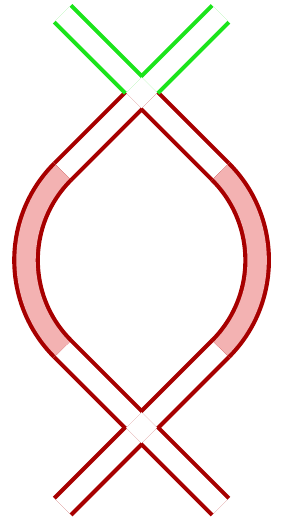}{412}\, \,\,+\,\, 
\includegraphicsd{.182}{Listones/Feo_RibbonsABAB_ABAB90} 
\,\,+ \,\,
\includegraphicsd{.182}{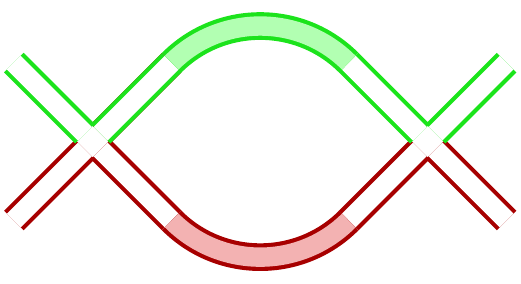}\,
\end{align}\label{wherebsquaresits}%
\end{subequations}%
It is also clear that the cyclic external leg structure of each of
these terms is $ABBA$.  Since it is easy to confuse the order of the
letters, we stress that this is already an extended version of the
$ABAB$-model to exemplify the one-loop structure of another coupling
constant. But expressions \eqref{wherebsquaresits}
also give the (by Claim \ref{thm:dos} unique) $\beta$-function where
the $b^2$ actually has to sit. \\

Adding the $ABBA$ operator modifies the flow (for, now, $\beta_b - b (1+2\eta ) \sim b c$) but in order to get the desired fixed points, higher-degree operators might still be required.
As pointed out in the paragraph before \cite[Sec. 3.3]{ABAB} when
addressing higher-degree operators, 
$\Tr[(AB)^3]$ is indeed forbidden. However, there are degree-six operators 
that do preserve the symmetries \eqref{symmetries}, concretely $\Tr (ABABAA)$ or $\Tr (ABABBB)$, and 
contribute to the $\beta_b$-function (see \cite[Thm. 7.2]{FRGEmultimatrix}, where
the RG-flow has been computed adding these operators), thus enriching the truncation.

\begin{remark} Some closing points one could learn from \cite{ABAB}:\label{poprawie_jesli}
\begin{itemize}

 \itemb Notice that if we could somehow make $ABAB$ indistinguishable from $ABBA$, then the $\beta_b$-function \cite[Eq. 3.41]{ABAB} would be, in that case, correct; see \eqref{obvious}.
 This happens for a sub-ensemble of pairs of Hermitian matrices $A,B$ such that $AB$ is Hermitian (for then, $A$ and $B$ commute). 
 \itemb  Notice that the graph 
 \begin{align}\label{withIsing}
\text{ $\,\includegraphicsd{.23}{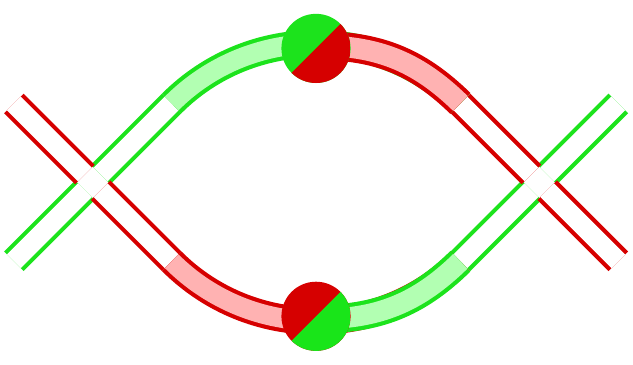}\,$}
\end{align}
would yield the $b^2$ term needed in \cite[\text{Eq. 3.41}]{ABAB}. However, this graph is not 
possible, since the Ising operator $ \sigma \Tr(AB)$, depicted with  
the bicolored bead and responsible for `changing color', is not in the truncation behind that equation; moreover, 
if added, it violates two symmetries in \eqref{symmetries}, on top of $b^2$ being screened by $\sigma^2$ (that graph is a one-loop containing four operators: two Ising, two $ABAB$, alternated).  
Nevertheless, it seems plausible that the contamination of the running operators (i.e. considering  operators that are not unitary invariant) might effectively lead to
a color-change, as in \eqref{withIsing}. 
\end{itemize}
 
\end{remark}

\section{Conclusion}
We showed that the connection that \cite{ABAB} 
established between the Functional Renormalization Group (FRG) and a phase
diagram---identified there with \cite[Fig. 4]{KazakovABAB}---relies
on a $\beta$-function that does not have the one-loop structure of the 
Functional Renormalization Equation. In
Section \ref{sec:proposal} above, we proposed to extend the
minimalist truncation of \cite{ABAB} in order to find a FRG-compatible set
of $\beta$-functions. Accomplishing this proposal would provide a
sound\footnote{This is not is not the same as `complete'. Analytic aspects
  should still be addressed, but these require more effort. 
  In order for this not to be in vain, it is a good idea to start from a correct combinatorics.} bridge, in the intention of \cite{ABAB}, between the FRG and Causal Dynamical Triangulations \cite{Ambjorn:2001br}
through the $ABAB$-model \cite{Ambjorn:ABAB,Ambjorn:ABAB2}.  Finally, we provided 
in Remark \ref{poprawie_jesli} the condition one would need to add in order for the 
$\beta_b$-function given by \cite[Eq. 3.41]{ABAB} to be correct.

  \section*{Acknowledgements}
The author was supported by the TEAM programme of the Foundation
for Polish Science co-financed by the European Union under the
European Regional Development Fund (POIR.04.04.00-00-5C55/17-00). 
 
 \bibliographystyle{JHEP}

\providecommand{\href}[2]{#2}\begingroup\raggedright\endgroup

\end{document}